\documentclass{article}

\usepackage[nonatbib,preprint]{neurips_2023}

\usepackage[utf8]{inputenc} 
\usepackage[T1]{fontenc}    
\usepackage{booktabs}       
\usepackage{amsfonts}       
\usepackage{nicefrac}       
\usepackage{microtype}      
\usepackage{xcolor}         
\usepackage{amsmath}
\usepackage{graphicx}
\usepackage{comment}

\title{Deep Generalized Green's Functions}

\author{%
  Rixi Peng \\
  Department of ECE\\
  Duke University\\
  Durham, NC 27708 \\
  \texttt{rixi.peng@duke.edu} \\
  \And
  Juncheng Dong \\
  Department of ECE\\
  Duke University\\
  Durham, NC 27708 \\
  \texttt{juncheng.dong@duke.edu} \\
  \And
  Jordan Malof \\
  Department of Computer Science\\
  University of Montana\\
  Missoula, MT 59812 \\
  \texttt{jordan.malof@umontana.edu} \\
  \And
  Willie J. Padilla \\
  Department of ECE\\
  Duke University\\
  Durham, NC 27708 \\
  \texttt{willie.padilla@duke.edu} \\
  \And
  Vahid Tarokh \\
  Department of ECE\\
  Duke University\\
  Durham, NC 27708 \\
  \texttt{vahid.tarokh@duke.edu} \\
}

\newtheorem{theorem}{Theorem}[section]

\newtheorem{lemma}[theorem]{Lemma}

\newtheorem{definition}[theorem]{Definition}

\newcommand{\lap}{\mathcal{L}} 
\newcommand{\argr}{\boldsymbol{r}}
\newcommand{\argxi}{\boldsymbol{\xi}}

\begin{document}

\maketitle

\begin{abstract}
In this study, we address the challenge of obtaining a Green's function operator for linear partial differential equations (PDEs). The Green's function is well-sought after due to its ability to directly map inputs to solutions, bypassing the need for common numerical methods such as finite difference and finite elements methods. However, obtaining an explicit form of the Green's function kernel for most PDEs has been a challenge due to the Dirac delta function singularity present. To address this issue, we propose the Deep Generalized Green's Function (DGGF) as an alternative, which can be solved for in an efficient and accurate manner using neural network models. The DGGF provides a more efficient and precise approach to solving linear PDEs while inheriting the reusability of the Green's function, and possessing additional desirable properties such as mesh-free operation and a small memory footprint. The DGGF is compared against a variety of state-of-the-art (SOTA) PDE solvers, including direct methods, namely physics-informed neural networks (PINNs), Green's function approaches such as networks for Gaussian approximation of the Dirac delta functions (GADD), and numerical Green's functions (NGFs). The performance of all methods is compared on four representative PDE categories, each with different combinations of dimensionality and domain shape. The results confirm the advantages of DGGFs, and benefits of Generalized Greens Functions as an novel alternative approach to solve PDEs without suffering from singularities. 

\end{abstract}

\section{Introduction}

Efficiently solving partial differential equations (PDEs) poses a significant challenge in the realm of scientific computing. These equations are ubiquitous, appearing in a multitude of disciplines ranging from physics and engineering to social science. However, the majority of PDEs cannot be explicitly solved, necessitating the deployment of numerical methods such as the finite difference method (FDM) or the finite element method (FEM). These methods operate by discretizing space and / or time, thus generating solutions on pre-established grids. The accuracy of these numerical solutions is contingent upon the preciseness of the discretization. Enhancing the grid density to achieve superior accuracy often implies a substantial escalation in computational resources, including memory, CPU time, and effort.

In order to mitigate these constraints, the development and implementation of alternative, mesh-free methods for PDEs is highly desirable. Recently, deep neural networks (DNNs) have emerged as a promising mesh-free solution capable of addressing a wide array of PDEs across various fields such as science, engineering, and mathematics. The most notable DNN-based PDE solving method is the physics-informed neural network (PINN) \cite{raissi2019physics}. In this approach, a DNN is used to represent a solution function which satisfies the PDE at every test point within the domain. Despite their potential benefits, a significant drawback of these DNN-based strategies lies in their training cost, which remains invariable for each unique PDE problem since a DNN model must be trained for each new problem . Furthemore, training a DNN to accurately solve a specific PDE problem can be a laborious process, dependent on the unique attributes of each problem as dictated by the differential operator, domain, boundary condition, and input function. 

\begin{figure}
  \centering
  \includegraphics[width=1\columnwidth]{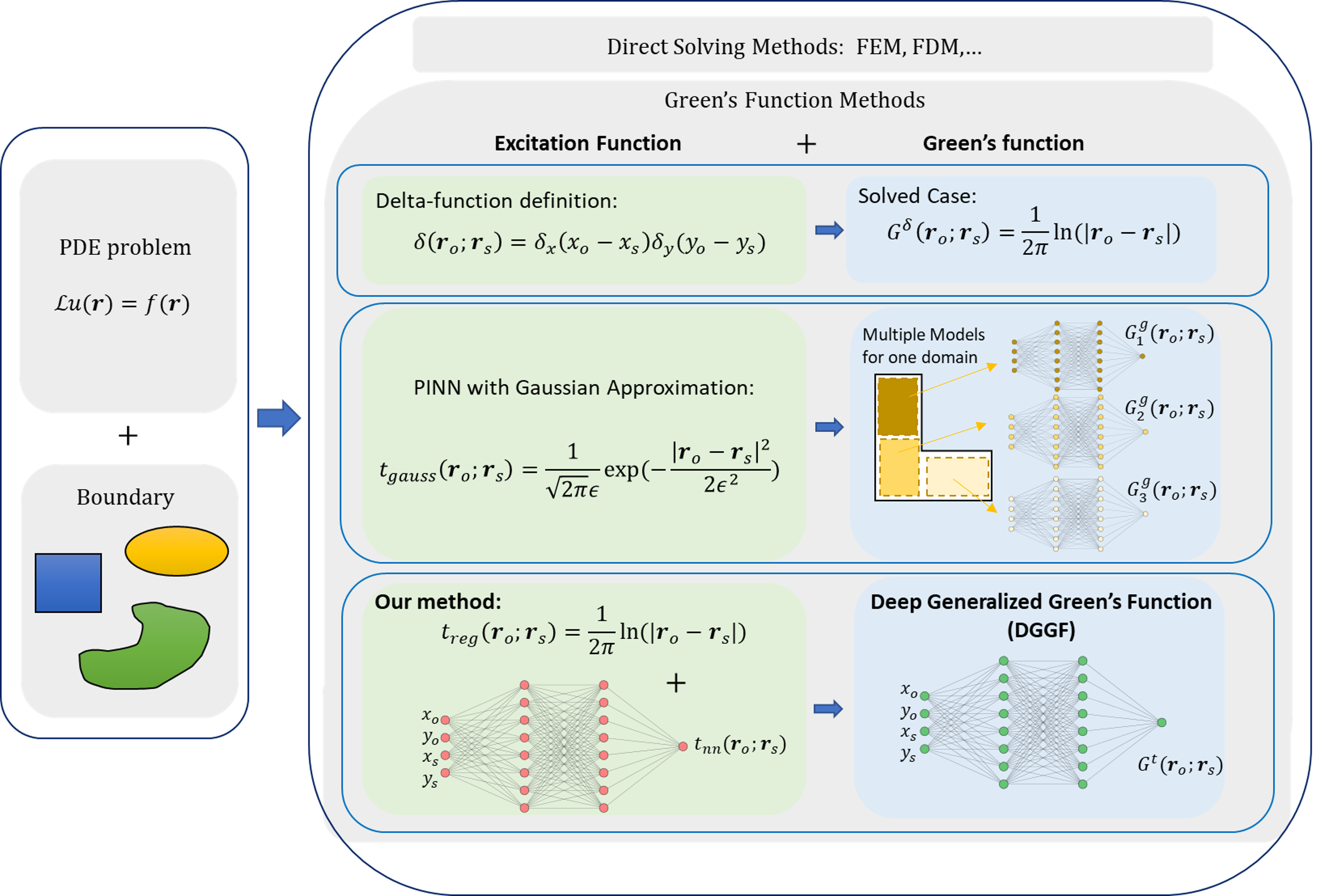}
  \caption{Schematic showing two types of PDE solvers: Direct Methods solve each PDE problem individually, and the Green's function method retrieves a kernel function first and constructs the solution in convolutional forms.}
  \label{fig:main_schematic}
\end{figure}


To address the aforementioned limitations of DNN-based PDE solutions, several strategies have been proposed to reduce their computational costs. Towards this goal, notable examples include the DeepRitz \cite{deepritz_2018}, which tackles the variational form of PDEs, and the Fourier Neural Operator (FNO) \cite{li2020fourier}, which leverages the Fourier transform methods widely used in many PDE problems. Another promising yet under-explored technique is the combination of Green's function methods with neural networks. Typically the Green's function is limited to linear PDEs and attempts to solve for a general integral kernel which is then used to construct a solution by convolving with the system input (or stimulus). The Green's function possesses several useful characteristics:
\begin{itemize}
\item Property (i): Regardless of the input function or boundary conditions, it transforms the process of solving the PDE into a convolution computation across the domain utilizing the Green's function as the kernel.
\item Property (ii): The Green's function naturally arises in boundary value and initial value PDE problems, and simplifies the problem by reducing the dimensionality.
\end{itemize}
Therefore, once a Green's function is constructed for a particular differential operator it is \textit{re-usable}, so that solutions for that operator can be constructed efficiently even if the input function and boundary conditions vary.

Despite their remarkable advantages, explicit Green's functions are only available for a limited number of scenarios, such as the Helmholtz equation in an unbounded homogeneous domain. In most cases, either the Green's function cannot be derived analytically or its resulting form becomes excessively intricate, thereby diminishing its utility; this has motivated the use of DNNs to approximate Green's functions.  Training DNNs for this purpose is challenging however, because it requires the derivatives of the network output to approximate a singularity function \cite{zemanian1987distribution}, namely a Dirac delta function. \cite{GaussApprox_2022}.   

Here we propose and demonstrate a generalized \emph{Green's function} method and the neural network representation, termed as \emph{Deep Generalized Green's functions} (DGGF), which addresses the singularity issue while inheriting all of the benefits of both using mesh-free neural network method and using the Green's function method to solve the linear PDE problems including Property (i) and (ii) above. Using a number of benchmark PDE problems, we demonstrate that the DGGF method provides stable solutions, exhibits a fast convergence rate, and achieves high accuracy.

Our contributions are summarized below: 
\begin{itemize}
    \item We construct Deep Generalized Green's Functions (DGGFs) that are applicable to all linear PDEs, and amenable to stable and efficient numerical computation. We prove that DGGFs maintain desirable Properties (i) and (ii) of the original Green's function,
    \item We demonstrate achieved accuracy in constructing PDE solutions using DGGFs for four different types of PDEs in various dimensions and for different boundary shapes and conditions, and
    \item We experimentally demonstrate that our proposal affords faster neural network training with superior performance compared to state of the art DNN based Green's function approaches.
\end{itemize}

\section{Related Work}
We formally define a general \textit{PDE problem} with a partial differential operator $\lap:\mathcal{U}\to\mathcal{U}^{*}$ of order $p>0$ from the solution function space $\mathcal{U}$ to its dual input function space $\mathcal{U}^{*}$ with support in a compact set $D\subset \mathbb{R}^{n}$. Let $\mathcal{U}$ be the Sobolev space of the same order of $p$. The boundary condition $\mathcal{B}$ specifies the solution function constraints on the domain boundary $\partial D$. Combining these two factors leads to,
\begin{align}
\begin{aligned}\label{eqn:def_pde_prob}
    &\lap u(\argr) = f(\argr), \quad \forall \argr \in D,\\
    &\mathcal{B}_{\argr}u(\argr)|_{\argr\in\partial D}=h(\argr),\\
\end{aligned}
\end{align}\label{eqn:pde}
\noindent where the input function $f\in \mathcal{U}^{*}: D\to \mathbb{R}$ is usually given and the solution function $u\in \mathcal{U}: D\to \mathbb{R}$ is to be solved. 

Many established PDE solvers can be chosen to solve the generic problem defined in Eq. \ref{eqn:def_pde_prob} based on the characteristics of the specific problem. Traditional methods include Fourier and Laplace Transform Methods, FEM and FDM etc. The idea of using a neural network to efficiently solve PDEs, to the best of our knowledge, dates back to the 1990s~\cite{chen_nn_pde_1996,Lagaris_1998,ann_solve_de_1990}. However, the lack of computational power and auto-differentiation capabilities severely limited the efficacy of neural networks in this domain. More recently, there has been a resurgence of the application of neural networks in solving PDEs due to works on physics-informed neural networks (PINNs)~\cite{PINN_struc_2022} in various scientific fields such as fluid dynamics~\cite{PINN_FD_2020,PINN_FD_2021}, thermodynamics~\cite{PINN_TD_2021}, and electromagnetism~\cite{PINN_EM_2021}. A detailed discussion of PINNs and their applications can be found in \cite{PINN_survey}. Several recent studies have proposed PINN improvements including reweighting loss terms in the objective function~\cite{PINN_reweighting_2021}, new loss terms~\cite{PINN_new_loss_2022}, and more complicated neural network structures~\cite{PINN_struc_2022}. Different DNN-based PDE solving methods are listed and compared in the Table \ref{tab:pde_table_1}. Another notable method for solving PDEs is DeepRitz~\cite{deepritz_2018}, which solves the variational form of the original strict differential form, resulting in relaxed constraints and often more robust convergence to the solution function. However, the limitation of applying this approach to calculating the Green's function is in their time-consuming process of evaluating a multidimensional integral at every epoch of the training process.

To the best of our knowledge, there have been a limited number of studies for approximating the Green's function with DNNs. A notable approach, referred to as GaussNet and described in ~\cite{GaussApprox_2022}, utilizes Gaussians to approximate the Dirac delta function. By employing smooth Gaussians, this method significantly relaxes the singularity of the Dirac delta function. The use of Gaussian functions offers a systematic approach to achieve high accuracy by reducing the width of the Gaussian. However, this method represents one Green's function on a single domain with multiple networks, making it challenging to apply to larger domains. Training may also be complex depending on the desired accuracy level, and approximation errors may be of concern.

The BI-GreenNet method is another important technique for expressing the non-singular part of Green's functions on different domains using DNNs \cite{lin2023bi}. This method takes advantage of the linearity of the Green's function and performs well when the singular part of the Green's function can be computed analytically. However, a limitation of this method is the need for an analytic solution for the singular part.

Various related data-driven approaches ~\cite{deepgreen_2021,green_human_2022,boulle2022data} learn the Green's function of an unknown system by incorporating it into a convolution integral and attempting to construct the solution function. These methods address a different type of problem, namely inferring the properties of an unknown system, instead of directly solving the underlying PDE. By utilizing the Green's function in a convolutional form, these approaches eliminate the need for evaluating the Dirac delta function. However, they may require fine sampling of the input function making them potentially challenging to implement in various scenarios of interest.

\begin{table}[h]
\caption{Properties of different PDE methods}
\label{tab:pde_table_1}
\vskip 0.15in
\begin{center}
\begin{small}
\begin{sc}
\begin{tabular}{cccc}
\toprule
Method & Fixed Grid & Re-usability & Data Requirement \\
\midrule
PINN & No & No & No\\
DGGF & No & Yes & No\\
GaussNet & No & Yes & No\\
BI-GreenNet & No & Yes & No\\
\bottomrule
\end{tabular}
\end{sc}
\end{small}
\end{center}
\vskip -0.1in
\end{table}

\section{Deep generalized Green's function}
\subsection{Preliminaries}


The Green's function corresponding to a \textit{linear} PDE problem described by Equation~\ref{eqn:def_pde_prob}, if it exists, is defined as the solution for a Dirac delta input function. To be precise, \textit{the Green's function problem} is formulated as, denoted by $G: D \times D \rightarrow \mathbb{R}$:
\begin{align}\label{eqn:green_pde_prob}
    \lap_{\argr}G(\argr,\argxi) &= \delta(\argr-\argxi), \, \forall \, \argr, \argxi \in D\\
    G(\argr,\argxi) &= 0, \,\, \,  \forall \, \argxi\in D^{\circ} \:\: \, \,  \forall\:\argr \in \partial D
\end{align}
where the subscript on $\lap_{\argr}$ denotes that the differential operator only operates on $\argr$ and $D^{\circ}$ denotes the interior of domain $D$. 
It should be noted that the Green's function must satisfy the same \textit{type} of boundary condition (Dirichlet, Neumann, or Robin) as that of the original PDE problem and should be in general zero on the boundary. The Green's function method then gives the solution to the problem as,
\begin{equation}
    u(\argr) = \int_{D}G(\argr,\argxi)f(\argxi)d\argxi + \int_{\partial D} h(\argr,\argxi')\frac{\partial_{\argr} G}{\partial \mathbf{n}}d\argxi',
\end{equation}
where $\mathbf{n}$ is the outward normal of the boundary ~\cite{kreyszig2008advanced}.

\subsection{Problem setting}
The goal of the Green's function is to find the operator which maps the input function $f$ to the solution function $u$. The formal description of solving the Green's function is outlined below. Let $\mathcal{U}$ be the Sobolev space of the order $p$ defined on a bounded domain $D\subset \mathbb{R}^{n}$. Let $\lap: \mathcal{U}\to \mathcal{U}^{*}$ be the linear partial differential operator such that $\lap u=f, u\in \mathcal{U}, f\in \mathcal{U}^{*}$. The corresponding Green's function operator $\mathcal{G}: \mathcal{U}^{*}\to \mathcal{U}$ can be formally written as,

\begin{align}\label{eqn:greens_operator}  
    u=(\mathcal{G}f)(\argr) = \int_{D}G(\argr,\argxi)f(\argxi)d\argxi,
\end{align}

where the kernel $G: D\times D\to \mathbb{R},$ is called the Green's function. 

The singular nature of the Dirac delta function poses challenges. 
Instead of solving the Green's function kernel with the formal definition, in this work, we assume $f$ is twice differentiable and subsequently we propose a generalized Green's function operator, 
\begin{align}
\mathcal{G}^{t}:  h \to u,\: \text{where} \:\:h=\Delta f
\end{align}

where the Laplacian operator $\Delta$ is given by$\Delta:=\sum_{i}\frac{\partial^{2}}{\partial x_i^2}$. 
With this generalized operator, an ordinary functions can replace the Dirac delta function in the optimization problem, formulated as 
\begin{align}\label{eqn:min_generalized}  
    \min_{\theta\in \mathbb{R}^{p}} \mathbb{E}_{\argr,\argxi \in  D}|\lap G^{\circ}_{\theta}(\argr, \argxi)-t(\argr,\argxi)|^2,
\end{align}
where $t(\argr,\argxi)$ is used in place of the Dirac delta function. We then derive an explicit form of the function $t(\argr,\argxi)$ and show that the resulting operator maps to the solution function of the PDE. The input function $t(\argr,\argxi)$ can be divided into the singular term $t_{s}(\argr,\argxi)$ and the regular term $t_{r}(\argr,\argxi)$. The singular term depends on the dimension $n$ of the problem while the regular term depends on the domain boundary condition. 
\begin{lemma}\label{lemma:exc}[Decomposition of the alternative input function]
Assume that both the differential operator $\lap$ is linear and the boundary condition is $t=0$. Then,
\begin{align}
    &\quad\quad t = t_{s} + t_{r} 
\end{align}
where $t_s$ and $t_r$ respectively satisfy:
\begin{align}
    &\Delta t_{s}(\argr,\argxi) = \delta (\argr-\argxi), \quad \lim_{\argr\to \infty}t_{s}(\argr,\argxi)  = 0, \\
    &\Delta t_{r}(\argr,\argxi) = 0, \quad t_{r}(\argr,\argxi)|_{\partial D} = -t_{s}(\argr,\argxi)|_{\partial D}.
\end{align}
\end{lemma}
Equation \ref{lemma:exc} itself forms a Green's function problem of the Poisson equation subject to the boundary condition. Next, we replace the Dirac delta function in the definition of Green's function in Eq.~\eqref{eqn:green_pde_prob} with $t(\argr,\argxi)$:
\begin{equation}\label{eqn:general_green_pde}
    \lap_{\argr}G^{t}(\argr,\argxi) = t(\argr,\argxi). 
\end{equation}
We term the corresponding solution $G^{t}$ as the \emph{generalized Green's function} subject to the input $t$.

In addition to Eq.~\ref{eqn:general_green_pde}, we also impose the following two requirements on $G^{t}$ on the boundary $\partial D$:  
\begin{align}
     \Delta_{\argr}G^{t}(\argr,\argxi)|_{\partial D}=0, \\
     \mathcal{B}G^{t}(\argr,\argxi)=0.
\end{align}

\begin{definition}[Generalized Green's Function]
The generalized Green's Function $G^t: D \times D \rightarrow \mathbb{R}$ on a compact domain is defined to be the solution to the following two connected PDE problems: 
\begin{align}
\Delta t(\argr,\argxi) &= \delta(\argr-\argxi),\\
\lap_{\argr}G^{t}(\argr,\argxi)& = t(\argr,\argxi),\\
\Delta_{\argr}G^{t}(\argr,\argxi)|_{\partial D}&=0,\\
\mathcal{B}G^{t}(\argr,\argxi)&=0.
\end{align}
\end{definition}

We note that the third condition can be used to derive the boundary condition for the alternative input function $t$. For instance, if the original problem is the Poisson equation, i.e., $\lap = \Delta$, the first condition on $G^{t}$ directly leads to $t(\argr,\argxi)=0$ on the boundary. For other types of operators, the boundary condition can be derived similarly. 

A generic analytic solution to the singular part $t_s(\argr,\argxi)$ for various domain dimensions is known with vanishing boundary condition, i.e., $\lim_{\argr\to\infty}t_{s}(\argr,\argxi)=0$. For $2$-dimensional domain, $t_{s}(\argr,\argxi) = ln(|\argr-\argxi|)$ and for $3$-dimensional domain, $t_{s}(\argr,\argxi) = 1/|\argr-\argxi|$. Since $t_{s}(\argr,\argxi)$ is expressed in the analytic form, the boundary value for the regular part $t_{r}(\argr,\argxi)$ can be easily calculated on arbitrarily shaped finite domain by $t_{r}(\argr,\argxi) = -t_{s}(\argr,\argxi)$. Then the problem reduces to solving for $t_r(\argr,\argxi)$ which is a standard boundary value problem in PDE that can be easily solved. Once the format of $t(\argr,\argxi)$ is fully determined via the procedures above, the Green's function problem is transformed to the generalized one with $t(\argr,\argxi)$ as the alternative input function, which is expressed exactly as normal functions.

\subsection{Model}\label{subsec:solving_generalized_greens}
In this section, we propose a unified paradigm for solving PDEs with neural networks and the generalized Green's function. Our method involves three steps which generate one auxiliary neural network and one primary neural network for the generalized Green's function. Following the concept of physics-informed neural network, we represent the solution function with parameterized neural networks while using the partial differential equations and boundary conditions to define the loss function which is minimized during training. The three steps are 
\begin{itemize}
    \item Step 1: Solve for the alternative input function $t$.
    \item Step 2: Solve the generalized Green's function $G^{t}$. 
    \item Step 3: Construct the solution using the generalized Green's function neural network.
\end{itemize}
It should be noted that step 1 and 2 are only executed once for a fixed PDE operator and domain $D$ but the trained neural network can be reused in step 3. We elaborate on each step in the following sections.
\subsubsection{Solving the alternative input function $t$.}
As stated in the Lemma \ref{lemma:exc}, the regular part $t_{r}(\argr,\argxi)$ is the solution to the boundary value problem defined with $\lap_{\argr} t_{r}(\argr,\argxi) = 0, t_{r}(\argr,\argxi) = g(\argr,\argxi)$ where the boundary values are determined by the analytic function of $t_{s}(\argr, \argxi)$. Suppose a parameterized neural network $\hat{t}_{r,\theta}(\argr,\argxi)$ is used to approximate $t_{r}(\argr,\argxi)$. For each training iteration, a batch of random interior point pairs of $(\argr_{i}\in D^{\circ}, \argxi_{i}\in D^{\circ}), i=1,2,...,N'_{dm}$ and a batch of random boundary point pairs $(\argr_{j}\in \partial D, \argxi_{j}\in D^{\circ}), j=1,2,...,N'_{bd}$ are independently sampled using the Latin hypercube sampling method. These sampled points across the domain or on the boundary are used to evaluate the residual loss or boundary loss, respectively. The total loss function at each iteration is defined as,
    \begin{align}\label{eqn:t_loss}
        L_{\theta}&=\frac{\lambda_{res}}{N'_{dm}}\sum_{i}^{N'_{dm}}L_{res}+\sum_{j}^{N'_{bd}}\lambda_{bd}L_{bd}\\
        &=\frac{\lambda_{res}}{N'_{dm}}\sum_{i}^{N_{dm}}|\lap\hat{t}_{\theta}(\argr_{i},\argxi_{i})|^{2} + \frac{\lambda_{bd}}{N'_{bd}}\sum_{j}^{N_{bd}}|\hat{t}_{r,\theta}(\argr_{j},\argxi_{j})-g(\argr_{j},\argxi_{j})|^2
    \end{align}\label{eqn:loss}
    
$\lambda_{res},\lambda_{db}$ are the weights to balance the magnitude between the PDE residual loss and the boundary loss. The partial derivatives in $\lap$ are all computed using the AutoGrad package in PyTorch. The stochastic gradient descent (SGD) method is used to minimize the loss and results in an accurate neural network representation of $\hat{t}_{r,\theta}(\argr,\argxi)$. 

\subsubsection{Solving for the generalized Green's function $G^{t}$.}
With the help of the auxiliary neural network model $\hat{t}_{r,\theta}(\argr,\argxi)$, the alternative excitation function $t(\argr,\argxi)$ can be readily evaluated at every point across the domain, i.e. $t(\argr, 
\argxi) = t_{s}(\argr,\argxi)+\hat{t}_{r,\theta}(\argr,\argxi)$. Now we use the same framework as above to train the generalized Green's function $\hat{G}^{t}_{\phi}$. For each training iteration, a batch of random interior point pairs of ${\argr_{i}\in D^{\circ}, \argxi_{i}\in D, i=1,2,...,N_{dm}}$ and a batch of random boundary point pairs ${(\argr_{j}\in \partial D, \argxi_{j}\in D^{\circ}), j=1,2,...,N_{bd}}$ are independently sampled using again the Latin hypercube sampling method. For this problem, the residual loss takes the alternative excitation function $t$ and sets the boundary values are zero, which leads to the definition of the total loss as,
\begin{align}\label{eqn:gt_loss}
    L_{\phi}&=\frac{\lambda_{res}}{N_{dm}}\sum_{i}^{N_{dm}}L_{res}+\frac{\lambda_{res}}{N_{bd}}\sum_{j}^{N_{bd}}\lambda_{bd}L_{bd}\\
    &=\frac{\lambda_{res}}{N_{dm}}\sum_{i}^{N_{dm}}|\lap\hat{G}^{t}_{\phi}(\argr_{i},\argxi_{i})-(t_{s}+\hat{t}_{r,\theta})(\argr_{i},\argxi_{i})|^{2} + \frac{\lambda_{res}}{N_{bd}}\sum_{j}^{N_{bd}}|\hat{G}^{t}_{\phi}(\argr_{j},\argxi_{j})|^2.
\end{align}

Once the generalized Green's function network finishes training, it can be used in Eq. (15) to determine the solution function for any excitation function $f$ with, 
\begin{equation}
    u(\argr_{0}) = \sum_{j} w_{j}\hat{G}^{t}_{\phi}(\argr_{0},\argxi_{j})\Delta f(\argxi_{j}).
\end{equation}
where $\argxi_{i,j}=1,2,...,n$ are the quadrature points for fast evaluation of the multi-dimensional integration, and $w_{j}$ are corresponding quadrature weights. It should be noted that the solution function values at multiple $\argr_{0}$ can be computed in parallel.

\section{Numerical experiments}
We demonstrate the performance of DGGF in solving PDE problems compared to several recent state-of-the-art baseline methods for different types of problems. We include different domains, including a square (SQ), a circle (CR), and two B-spline curve enclosed loops (B1 \& B2), and two 3D boundaries including a cube with 1/8 corner cut (CC), and an ellipsoid (EP), along with four different PDE types, listed in Table \ref{tab:pde_table_2}. The exact domain shapes are illustrated in Fig. \ref{fig:domains}.

\begin{figure}
  \centering
  \includegraphics[width=0.5\columnwidth]{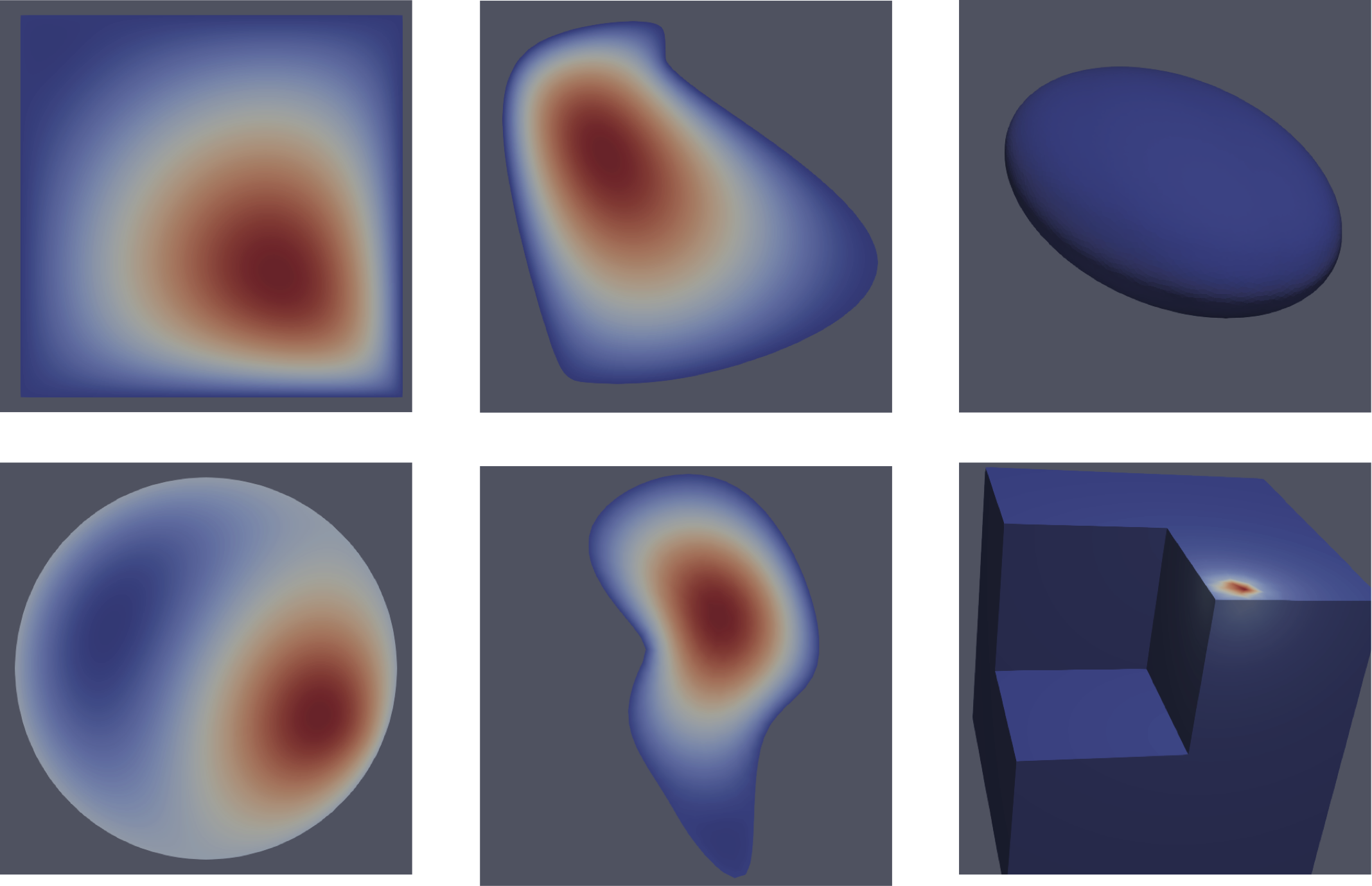}
  \caption{Different types of domain boundaries explored in the experiments. The color map indicates the generalized Green's function of a given fixed source location for the Poisson equation.}
\end{figure}\label{fig:domains}

Specifically, we demonstrate that DGGF has higher accuracy compared to GaussNet and constructs faster results compared to PINN on several classes of widely-studied PDEs.  We present results of different dimensions and boundary shapes in our experiments. For each type of the PDE problems, we are interested in the performance of DGGF compared to the following baseline models:
\begin{itemize}
    \item Method (I): Gaussian Approximation of the Dirac delta function (GaussNet),
    \item Method (II): Physics-informed neural network (PINN),
    \item Method (III): Numeric Green's Function (NGF),
\end{itemize}

\begin{table}[ht]
\caption{PDE problems explored in the experiments. Boundary shapes include square (SQ), circular (CR), two Jordan curves B-spline 1 (B1) and B-spline 2 (B2), corner cut cube (CC), and an ellipsoid (EP).
}
\label{tab:pde_table_2}
\vskip 0.15in
\begin{center}
\begin{small}
\begin{sc}
\begin{tabular}{cccc}
\toprule
Class & Expression & Boundary Shape \\
\midrule
Poisson & $\sum_{x,y,z}\partial^{2}_{x,y,z}$ & SQ, CR, B1, B2, CC, EP  \\
Helmholtz & $\sum_{x,y,z}\partial^{2}_{x,y,z}+k^{2}$ & SQ, CR, B1, B2, CC, EP  \\
Heat & $\sum_{x,y,z}\partial^{2}_{x,y,z}-\partial_{t}$ & SQ, CR, B1, B2, CC, EP \\
Klein-Gordon & $\sum_{x,y,z}\partial^{2}_{x,y,z}-\partial_{t}^2 -k^{2}$ & SQ, CR, B1, B2, CC, EP \\
\bottomrule
\end{tabular}
\end{sc}
\end{small}
\end{center}
\vskip -0.1in
\end{table}

We use the FEM solver FEniCSx ~\cite{AlnaesEtal2015,ScroggsEtal2022,BasixJoss,LoggEtal2012,LoggWells2010,LoggEtal_10_2012,KirbyLogg2006,LoggEtal_11_2012,OlgaardWells2010,AlnaesEtal2014,kirby2010} to obtain the solution close to the ground truth for all the experiments using very fine meshes.   

\begin{figure}[pbh!]
  \centering
  \includegraphics[width=0.75\columnwidth]{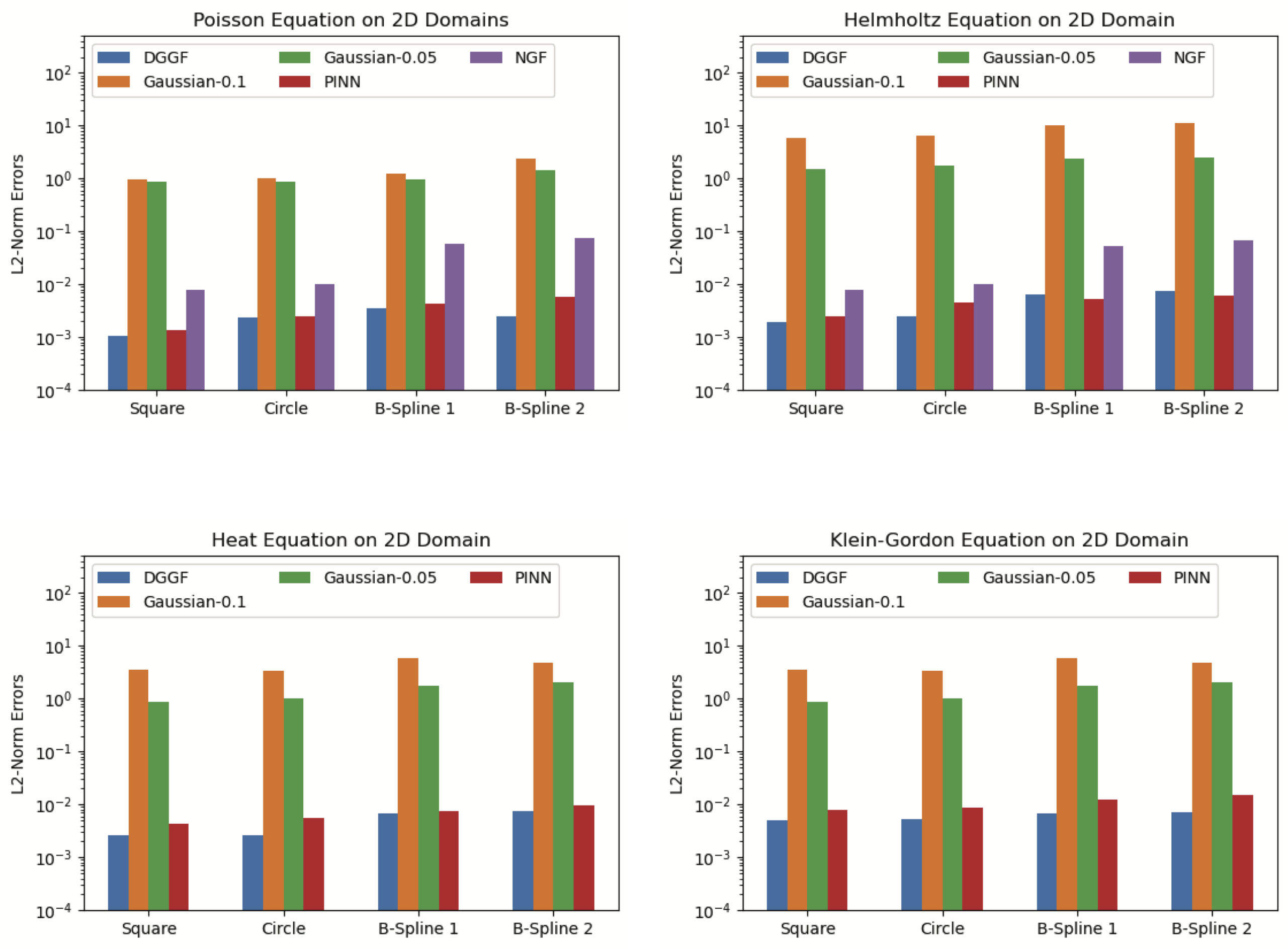}
  \caption{Results of 2D experiments with different PDEs and boundaries showing the mean L2 norm error, including the DFFG (our approach shown as the blue colors), Gaussian (0.05 green colors) and (0.1 shown as the orange colors), PINN (purple), and NGF (red). }
\end{figure}\label{fig:res_2d}

\subsection{Model training}
For our method, the PINN, and the GaussNet, we use the exact same DNN structure and training procedures throughout for all experiments for fair comparison. The DNN comprises $8$ hidden layers with Rectified Linear Units (ReLU) activation functions, and $100$ neurons in each layer. Note that the GaussNet method prescribes using multiple networks, instead of a single network, to represent the fast-varying Green's function on the whole domain. We tested this multi-network approach on simple 2D domain problems with 16 networks and results are included in the Appendix. The complexity of arbitrary segmentation of irregular domains hinders the implementation of this method on all the cases explored in this study. The hyperparameters involved in the training of the models are $N_{dm}$, $N_{bd}$, $N'_{dm}$, $N'_{bd}$,$\lambda_{res}$, $\lambda_{bd}$. Please see Sec. \ref{subsec:solving_generalized_greens} for interpretation of these hyperparameters. For simplicity, we always set $\lambda_{res}=1$ and search $\lambda_{bd} \in [1, 5, 50]$, $N_{dm} \in [50000, 100000, 150000, 200000]$, $N_{bd} \in [25000, 50000, 100000, 75000, 150000, 200000]$. The lists of values for $N'_{dm}$ and $N'_{bd}$ are the same as that of $N_{dm}$ and $N_{bd}$. We use the learning rate of $1\times 10^{-3}$ with ADAM~\cite{adam} and zero regularization. The GaussNet method requires one additional hyperparameter, which is the width of the Gaussian used to approximate the Dirac delta function. For this, we tested two values $\epsilon \in \{.05, 0.1\}$, for every problem considered below. 

\begin{figure}[ptb!]
  \centering
  \includegraphics[width=0.75\columnwidth]{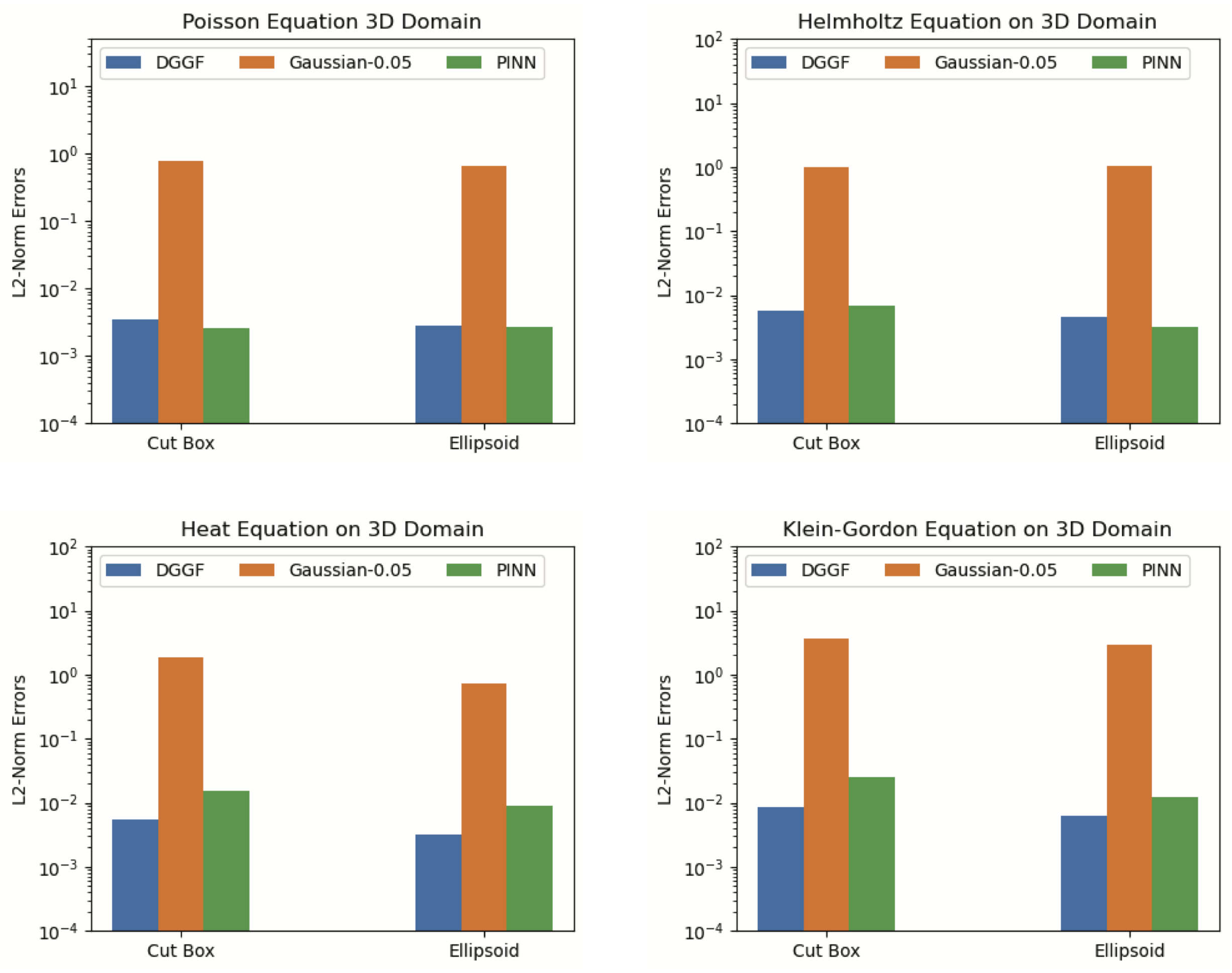}
  \caption{Results of 3D experiments with different PDEs and boundaries showing the mean L2 norm error, including the DFFG (our approach shown as the blue colors), Gaussian (0.05 green colors) and (0.1 shown as the orange colors), and PINN (red).}
\end{figure}\label{fig:res_3d}

\subsection{Results}
\subsubsection{Accuracy}
A comparison of the accuracy between DGGF and baseline methods are present in Figure~\ref{fig:res_2d}. It can be observed that the DGGF outperforms both the GaussNet for both choices of the Gaussian width explored with the single network representation. In particular, the DGGF realizes at least three orders of magnitude smaller errors than that of the GuassNet, and is comparable to the PINN method. These results indicate that by transforming the Green's function to the generalized Green's function, the neural network model can learn a better kernel to construct the solution functions. Despite comparable accuracy, our method is 20,000 times faster than the PINN in the solution function construction step, with the parallel convolution in our method taking ~$0.08$ s for 2D domains, compared to a $~30$ min training time for the PINN method on any single 2D problem. 

\begin{figure}[ptb!]
\begin{center}
(a){\includegraphics[width=0.3\columnwidth]{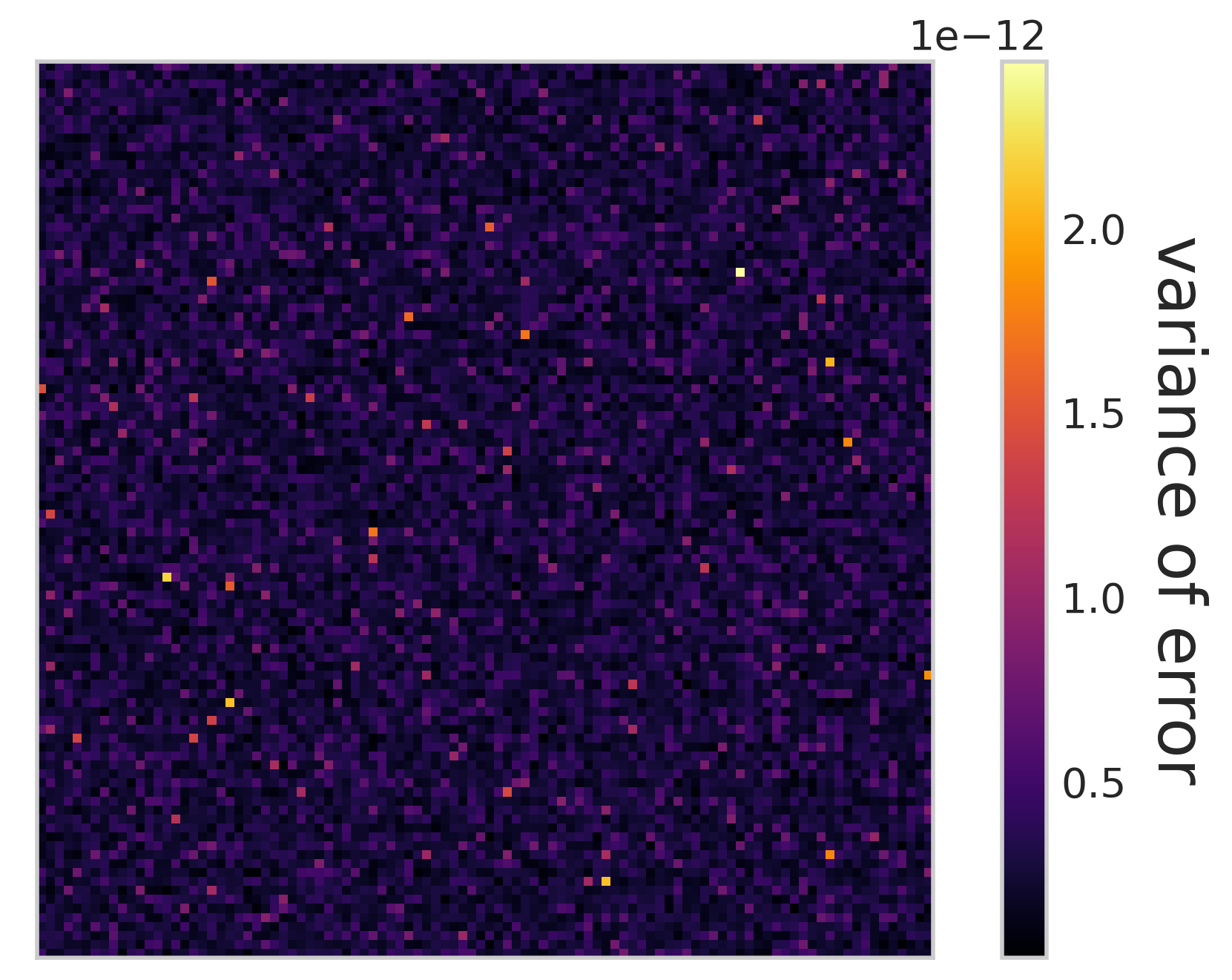}}
(b){\includegraphics[width=0.3\columnwidth]{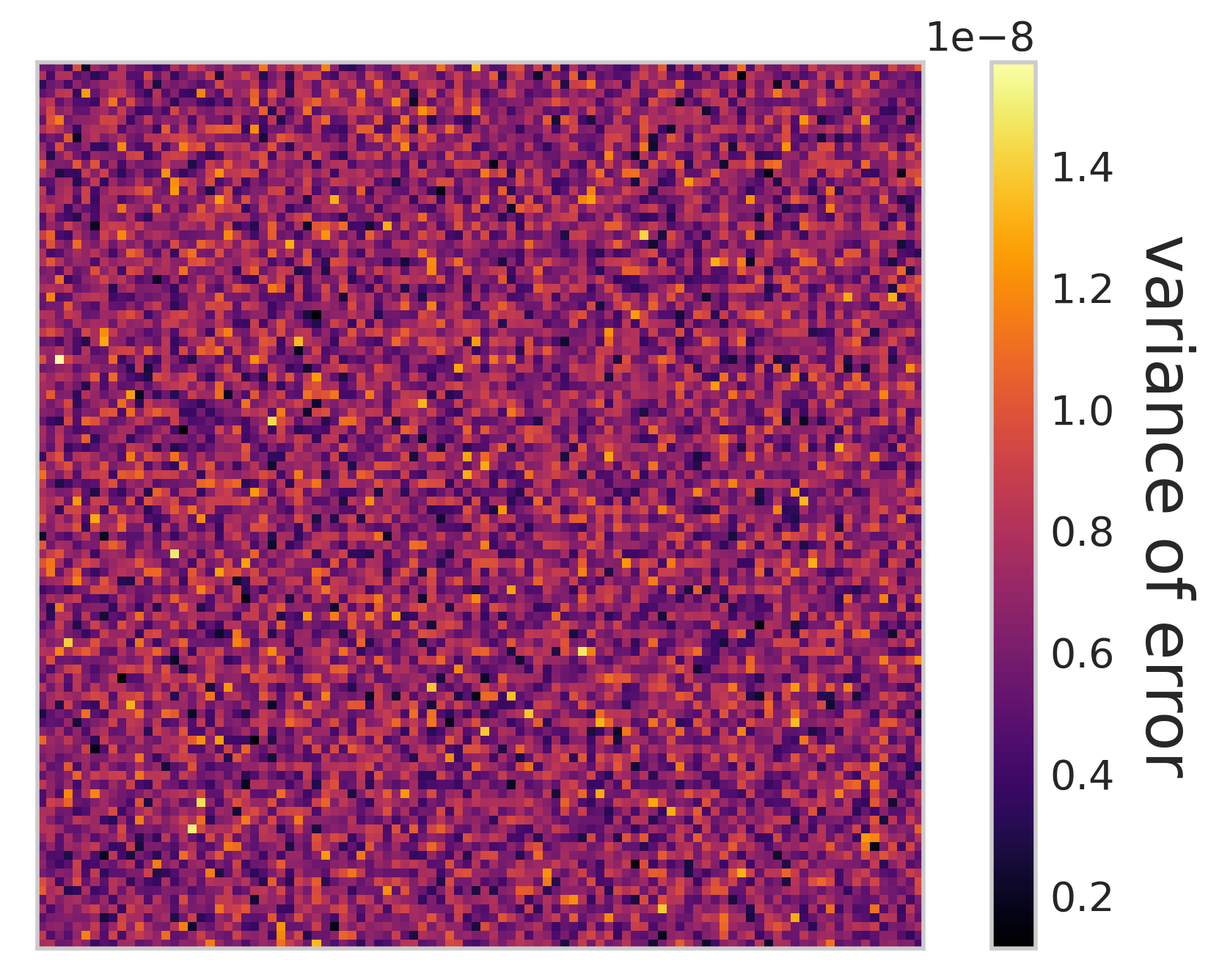}}
\caption{Stability of DGGF. Here we perform only Step 2 ($G^t$) in (a) and Step 1 + Step 2 (b). The variance of accuracies across points in the domain (due to random initialization of neural networks) is used to measure the stability of DGGF. Shown in (a) and (b) are variances of accuracies for each point in the 2-d Box domain. }
\label{fig:ablation_stability}
\end{center}
\end{figure}


\subsection{Ablation Study}
\subsubsection{Stability}
We use $2$-d Poisson problem with Box boundary to study the stability of DGGF, i.e., the variance of error caused by different random initializations of the neural networks. We study two scenarios: (1) only training the networks in step 2 of DGGF and (2) training networks in both step 1 and step 2. Specifically, in both scenarios, we use network of $8$ layers, all with $100$ neurons. We set $N_{dm}=100000$, $N_{bd}=300000$ and $\lambda_{bd} = 5$. Then we train $30$ neural networks with different initializations. The variances of their accuracies for estimating Green's Function at various points in domain are then computed. As can be seen in Figure~\ref{fig:ablation_stability}, the variance of accuracy across the domain is uniformly small. In particular, the variance of only training the network for $G^t$ is to the order of $10^{-11}$, demonstrating the stability of DGGF.   
 
\section{Conclusions}
To harness the full potential of Green's function reusability, we introduce the concept of deep generalized Green's function. By circumventing the singularity associated with the Dirac delta function while preserving theoretical accuracy, our method offers a practical solution. Empirical tests conducted on a comprehensive selection of partial differential equations (PDEs) have confirmed the effectiveness of our approach. Notably, DGGF exhibits lower computational resource requirements, faster convergence, and, crucially, enables efficient reuse for solving PDEs of the same type.



\bibliographystyle{unsrt}

\appendix

\section{Gaussian Approximation with Segmented Domains}
In the main text, one neural network $G_{\theta}(\argr,\argxi)$ is used to represent the Green's function of some PDE problem for all $\argxi$ on the interior domain. In this section, we aim to directly compare the GaussNet scheme proposed in  \cite{GaussApprox_2022} which uses multiple neural network models, to our method which uses a single model. In more detail, for the GaussNet method the entire interior domain is segmented into several subdomains where each is represented by a neural network $G_{\theta,i}(\argr,\argxi), \forall \argr \in D, \forall \argxi \in D_{i}, \bigcup_{i}D_{i}=D^{\circ}$. It should be noted that using multiple networks should decrease the representation error especially on edge areas of the domain and potentially decrease the training difficulty. However, it cannot solve the approximation error inherent in the Gaussian approximation itself, which is solely determined by the Gaussian width. Furthermore, using multiple networks for one single domain will drastically increase the computational burden for training, as well as the storage requirement for each domain. Here, we used a Gaussian width of $s=0.02$ and $6x6$ squares to segment three different domains, (Square, Annulus, and L-Shape), following the scheme in \cite{GaussApprox_2022}. We study the Poisson equation subject to the second order polynomial input functions $f(x,y) = a_{1}x^2+a_{2}xy+a_{3}y^2+a_{4}x+a_{5}y+a_{6}$ for 10 different random realizations $a_{i}\sim\mathcal{N}(0,2)$, and for $i=1,...,6$. The ground truth is computed using FEM with a dense mesh for these domains. 

\begin{table}[h]
\caption{Averaged Numerical Error of Multiple Network Models and DGGF on 2D Domains}
\label{tab:pde_table_3}
\vskip 0.15in
\begin{center}
\begin{small}
\begin{sc}
\begin{tabular}{cccc}
\toprule
Method & Square & Annulus & L-shape \\
\midrule
GaussNet & 3.37e-2(6x6) & 9.97e-2 (6x6) & 1.43e-2(6x6)\\
GaussNet w/ Symmetry Loss & 2.76e-2(6x6) & 7.45e-2 (6x6) & 1.06e-2(6x6)\\
DGGF & 1.13e-3 & 3.47e-3 & 1.42e-3\\
\bottomrule
\end{tabular}
\end{sc}
\end{small}
\end{center}
\vskip -0.1in
\end{table}

The symmetry loss, introduced in \cite{GaussApprox_2022}, is defined as $\min|G_{\theta}(\argr,\argxi)-G_{\theta}(\argxi,\argr)|$. However, this symmetry property is valid for reciprocal systems and may not hold generally. For this reason, we did not include this symmetry constraint in the training of DGGFs, which allows our method to generalize to a wider range of PDE problems. The average run-time for these experiments is given in Table \ref{tab:pde_table_4}.

\begin{table}[h]
\caption{Training Computational Complexity for The Experiments Presented in Table \ref{tab:pde_methods_1} in GPU hours}
\label{tab:pde_table_4}
\vskip 0.15in
\begin{center}
\begin{small}
\begin{sc}
\begin{tabular}{cccc}
\toprule
Method & Square & Annulus & L-shape \\
\midrule
GaussNet &  1.89h x 36& 2.48h x 36 & 3.04h x 27\\
GaussNet w/ Symmetry Loss &  1.53h x 36 & 2.12h x 36 & 2.34h x 27\\
DGGF & 0.12h + 0.08h & 0.42h + 0.46h & 0.12h + 0.13h\\
\bottomrule
\end{tabular}
\end{sc}
\end{small}
\end{center}
\vskip -0.1in
\end{table}

For GaussNet, the averaged computation time is reported since all the neural networks for subdomains are trained in parallel. Note that for the L-shaped domain, although it is segmented into $6\times6$, this domain only fills out $27$ subdomains. For DGGF, there are two seriel steps of training two neural networks and the computation time is reported as the sum in the table. Since the GaussNet uses LBFGS for optimization, it is in general more time-consuming for each epoch compared to ADAM. In DGGF, the relaxed singularity means the generalized Green's function vary slower on the domain than the original Green's function, making it converge fast with ADAM. 

\section{Deriving the Generalized Green's Function}

In this section, we demonstrate that the generalized Green's function solved above can also construct the solution via integral operation with the input function $f$.
\begin{theorem}[Equivalence between Generalized Green's Function and Green's Function]
We have
\begin{equation}
    u(\argr) = \int_{D} G^{t}(\argr,\argxi)\Delta f(\argxi) d\argxi+ \int_{\partial D}f\frac{\partial G^{t}}{\partial n} - G^{t}\frac{\partial f}{\partial n}ds
\end{equation} 

This equation can be further simplified if $G^{t}, f$ or their partial derivatives vanish on the boundary. In that case, the relationship of the generalized Green's function and the solution function is simply,
\begin{equation}
    u(\argr) = \int_{D} G^{t}(\argr,\argxi)\Delta f(\argxi) d\argxi
\end{equation}
In other words, the solution can be expressed as the integral of the generalized function and the Laplacian of the input functions. 
\end{theorem}

The problem to be solved is to find the solution function $u$ given the input function $f$ in Eq.(1). We start with formulating a Green's identity integral of $G^{t}$ and $\lap u$ with the operator of Laplacian, 
\begin{equation}
    \int_{D}d\argr G^{t}\Delta(\lap u)-(\lap u)\Delta G^{t}  = \int_{\partial D}dsG^{t}\frac{\partial\lap u}{\partial n}-\lap u \frac{\partial G^{t}}{\partial n}
\end{equation}
Assume $\lap$ to be self-adjoint and apply again the Green's identity on it, we could get
\begin{equation}
    \int_{D} \Delta G^{t}(\lap u)-u\Delta(\lap G^{t})d\argr = \int_{\partial D}ds\gamma_{bd}
\end{equation}
where interchangeability of $\lap$ and the Dirac delta function is assumed and the boundary terms are involved with $\Delta G^{t}, u$ and their partial derivatives on the boundary depending on $\lap$. Since we have specified $\Delta G^{t}=0$ on the boundary, the boundary terms in the above equation vanish and $\int \lap u\Delta G^{t}=\int u\lap\Delta G^{t} = \int u\Delta t = \int u\Delta = u$. Therefore, by inserting this equation into Eq.(12), we have
\begin{equation}
    u(\argr) = \int_{D} G^{t}(\argr,\argxi)\Delta f(\argxi) d\argxi + \int_{\partial D}f\frac{\partial G^{t}}{\partial n}-G^{t}\frac{\partial f}{\partial n}ds
\end{equation}

\end{document}